\documentclass[a4paper]{article}

\usepackage{amsthm}
\usepackage{authblk}
\usepackage{amssymb}
\usepackage{amsmath}
\usepackage{a4wide}
\usepackage{graphics}
\usepackage{calc}
\usepackage{microtype}
\usepackage{tikz}
\usepackage{thm-restate}

\usepackage[boxed, ruled,noend,linesnumbered]{algorithm2e}

\usepackage{wrapfig}
\usepackage{cite}

\newtheorem{theorem}{Theorem}

\newtheorem{lemma}{Lemma}
\newtheorem{conjecture}{Conjecture}
\newtheorem{observation}{Observation}

\newcounter{linenumber}

\def\D{\ensuremath{\mathcal{D}}}

\def\R{\ensuremath{\mathcal{R}}}

\def\E{\ensuremath{\mathcal{E}}}

\newcommand{\remove}[1]{}

\def\E {\mathcal{E}}

\def\Write{\ensuremath{\textit{Write}}}

\def\Update{\ensuremath{\textit{Update}}}
\def\Snapshot{\ensuremath{\textit{Snapshot}}}

\def\Collect{\ensuremath{\textit{Collect}}}

\newcommand{\ignore}[1]{}

\title{Progress-Space Tradeoffs in Single-Writer Memory Implementations}

\author[1,2]{Damien Imbs$^*$}
\author[3]{Petr Kuznetsov$^*$}
\author[3]{Thibault Rieutord\footnote{This work has been supported by the Franco-German DFG-ANR Project DISCMAT (14-CE35-0010-02) devoted to connections between mathematics and distributed computing.}}

\affil[1]{LIF, Aix-Marseille Universit\'e \& CNRS, France}
\affil[2]{Bremen University, Germany
  
  \texttt{damien.imbs@lif.univ-mrs.fr}}
\affil[3]{LTCI, T\'el\'ecom ParisTech, Universit\'e Paris Saclay, France

  \texttt{firstname.lastname@telecom-paristech.fr}}
  
\date{}

\begin{document}

\maketitle

\begin{abstract}
Many algorithms designed for shared-memory distributed systems assume
the \emph{single-writer multi-reader} (SWMR) setting where 
each process is provided with a unique register 
that can only be written by the process and read by all. 
In a system where computation is performed by a bounded number $n$ of processes 
coming from a large (possibly unbounded) set of potential participants, 
the assumption of an SWMR memory is no longer reasonable. 
If only a bounded number of multi-writer multi-reader (MWMR) registers are provided, 
we cannot rely on an \emph{a priori} assignment of processes to registers. 
In this setting, implementing an SWMR memory, or equivalently,
ensuring \emph{stable writes} 
(i.e., every written value persists in the memory), is desirable.

In this paper, we propose an SWMR implementation that adapts
the number of MWMR registers used to the desired progress condition.
For any given $k$ from $1$ to $n$, we present an algorithm that uses
$n+k-1$ registers to implement a \emph{$k$-lock-free} SWMR memory.
In the special case of $2$-lock-freedom, 
we also give a matching lower bound of $n+1$ registers, 
which supports our conjecture that the algorithm is space-optimal.
Our lower bound holds for the strictly weaker progress condition of
$2$-obstruction-freedom, which suggests that the space complexity for
$k$-obstruction-free and $k$-lock-free SWMR implementations might coincide. 
\end{abstract}

\section{Introduction}

We consider a distributed computing model in which 
at most $n$ \emph{participating} processes communicate 
via reading and writing to a shared memory.
The participating processes come from a possibly unbounded set 
of \emph{potential} participants: each process has a unique identifier
(IP address, RFID, MAC address, etc.) which we, 
without loss of generality, assume to be an integer value.
Given that processes do not have an \emph{a priori} knowledge 
of the participating set, it is natural to assume that they can only 
\emph{compare} their identifiers to establish their relative order, 
otherwise they essentially run the same algorithm~\cite{Ram30}.
This model is therefore called \emph{comparison-based}~\cite{ABDPR90}. 
In the comparison-based model with bounded shared memory, 
we cannot assume that the processes are provided with 
a prior assignment of processes to distinct registers.
The only suitable assumption, as is the case for anonymous systems~\cite{Y16}, 
is that processes have access to \emph{multi-writer multi-reader} registers~(MWMR).

In this paper, we study the \emph{space complexity} of comparison-based implementations 
of an \emph{abstract} single-writer multi-reader (SWMR) memory.
The abstract SWMR memory allows each participating process to \emph{write} 
to a private abstract memory location and to \emph{read} 
from the abstract memory locations of participating processes.
The SWMR abstraction can be further used to build higher-level
abstractions, such as renaming~\cite{ABDPR90} and atomic snapshot~\cite{AAD93}.

To implement an SWMR memory, we need to  ensure that every write performed 
by a participating process on its abstract SWMR register is \emph{persistent}: 
every future abstract read must see the written value, 
as long as it has not been replace by a more recent \emph{persistent} write.
To achieve persistence in a MWMR system, the emulated abstract write 
may have to update \emph{multiple} base MWMR registers in order to ensure
that its value is not overwritten by other processes.
A natural question arises: \emph{How many base MWMR registers do we need?}

In this paper, we show that the answer depends on the desired progress condition.
It is immediate that $n$ registers are required for a \emph{lock-free} implementation, 
i.e., we want to ensure that at least \emph{one} correct process makes progress.
Indeed, any algorithm using $n-1$ or less registers can be brought into the situation 
where \emph{every} base register is \emph{covered}, i.e., 
a process is about to execute a write operation on it~\cite{BL83}.
If we let the remaining process $p_i$ complete a new abstract write operation,
the other $n-1$ processes may destroy the written value by making 
a \emph{block write} on the covered registers 
(each covering process performs its pending write operation).
Thus, the value written by $p_i$ is ``lost'': no future read would find it. 
It has been recently shown that $n$ base registers are not only necessary, 
but also sufficient for a lock-free implementation~\cite{DFGR15}.

A \emph{wait-free} SWMR memory implementation that guarantees progress 
to \emph{every} correct process can be achieved with $2n-1$ registers~\cite{DFGR15}. 
The two extremes, lock-freedom and wait-freedom, suggest an intriguing question: 
is there a dependency between the amount of progress the implementation provides 
and its space complexity: if processes are guaranteed more progress, 
do they need more base registers?

\subparagraph*{Contributions.}
In this paper, we give an evidence of such a dependency.
Using novel covering-based arguments, we show that 
any \emph{$2$-obstruction-free} algorithm requires $n+1$ base MWMR registers.
Recall that $k$-obstruction-freedom requires that every correct process 
makes progress under the condition that at most $k$ processes are correct~\cite{Tau09}.
The stronger property of \emph{$k$-lock-freedom}~\cite{BG15} additionally guarantees that 
if more than $k$ processes are correct, then at least $k$ out of them make progress. 

We also provide, for any $k=1,\ldots, n$, a $k$-lock-free SWMR memory implementation 
that uses only $n+k-1$ base registers. 
Our lower bound and the algorithm suggest the following:

\begin{conjecture}
It is impossible to implement a $k$-obstruction-free SWMR memory 
in the $n$-process comparison-based model using $n+k-2$ MWMR registers.
\end{conjecture}

An interesting implication of our results is that $2$-lock-free and 
$2$-obstruction-free SWMR implementations have the same optimal space complexity.
Given that $n$-obstruction-freedom and $n$-lock-freedom coincide with wait-freedom,
we expect that, for all $k=1,\ldots,n$, \mbox{$k$-obstruction-free} and $k$-lock-free 
(and all progress conditions in between~\cite{BG15}) 
require the same number $n+k-1$ of base MWMR registers.
Curiously, our results highlight a contrast between complexity and computability, 
as we know that certain problems, e.g., consensus, 
can be solved in an obstruction-free way, but not in a lock-free way~\cite{HLM03}.

\subparagraph*{Related work.}
\label{sec:related}
Jayanti, Tan and Toueg~\cite{JTT96} gave linear lower bounds 
on the space complexity of implementing a large class 
of \emph{perturbable} objects (such as CAS and counters).
For atomic-snapshot algorithms, Fatourou, Ellen and Ruppert~\cite{FER06} 
showed that there is a tradeoff between the time and space complexities, 
both in the anonymous and the non-anonymous cases. 
Zhu~\cite{Zhu16} showed that $n-1$ MWMR registers are required 
for obstruction-free consensus. 

Delporte et al.~\cite{DFKR15} studied the space complexity 
of anonymous $k$-set agreement using MWMR registers, 
and showed a dependency between space complexity and progress conditions. 
In particular, they provide a lower bound of $n-k+m$ MWMR registers 
to solve anonymous \emph{repeated} $k$-set agreement 
in the $m$-obstruction-free way, for $k<m$.
Delporte et al.~\cite{DFGR13} showed that obstruction-free $k$-set agreement 
can be solved in the $n$-process comparison-based model using $2(n-k) + 1$ registers.
This upper bound was later improved to $n-k+m$ for the progress condition of 
$m$-obstruction-freedom ($m\leq k$) by Bouzid, Raynal and Sutra~\cite{BRS15}. 
In particular, their algorithm uses less than $n$ registers when~$m<k$. 
 
To our knowledge, the only lower bound on the space-complexity 
of implementing an SWMR memory has been given by Delporte et al.~\cite{DFGR15}
who showed that lock-free comparison-based implementations require $n$ registers.

Delporte et al.~\cite{DFGR15} proposed two SWMR memory implementations:
a lock-free one, using~$n$ registers, and a wait-free one, using $2n-1$ registers. 
These algorithms are used in~\cite{DFGL13} to implement a \emph{uniform} SWMR memory, 
i.e., assuming no prior knowledge on the number of participating processes.
Assuming that $p$ processes participate, the algorithms use $3p+1$ 
and $4p$ registers for, respectively, lock-freedom and wait-freedom.

\subparagraph*{Roadmap.}
The paper is organized as follows.
Section~\ref{sec:prel} defines the system model and states the problem.
Section~\ref{sec:alg} presents a $k$-lock-free SWMR memory implementation.
Section~\ref{sec:lb} shows that a $2$-obstruction-free SWMR memory implementation 
requires $n+1$ MWMR registers and hence that our algorithm is optimal for $k=2$.
Section~\ref{sec:disc} concludes the paper with implications and open questions.

\section{Model}
\label{sec:prel}

We consider the asynchronous shared-memory model, in which a bounded number $n>1$ 
of asynchronous crash-prone processes communicate by applying read and 
write operations to a bounded number $m$ of \emph{base} atomic multi-writer
multi-reader atomic registers.
An atomic register $i$ can be accessed with two memory operations: 
$\textit{write}(i,v)$ that replaces the content of the register with value $v$, 
and $\textit{read}(i)$ that returns its content.
The processes are provided with unique identifiers from an unbounded name space.
Without loss of generality, we assume that the name space is the set of positive integers.

\subsection{States, configurations and executions}

An algorithm assigned to each process is a (possibly non-deterministic) automaton
that accepts high-level operation requests as an application input.
In each state, the process is poised to perform a \emph{step}, i.e., 
a read or write operations on base registers. 
Once the step is performed, the process changes its state
according to the result the step operation, possibly non-deterministically and 
possibly to a step corresponding to another high-level operation.

A \emph{configuration}, or system state, consists of 
the state of all processes and the content of all MWMR registers.
In the \emph{initial} configurations, all processes are in their
initial states, and all registers carry initial values.

We say that a step $e$ by a process $p$ is \emph{applicable} to a configuration $C$, 
if $e$ is the pending step of $p$ in $C$, and we denote $Ce$ the configuration 
reached from $C$ after $p$ performed~$e$. A sequence of steps $e_1,e_2,\ldots$ 
is applicable to $C$, if $e_1$ is applicable to $C$, $e_2$ is applicable to~$Ce_1$, etc. 
A (possibly infinite) sequence of steps applicable to a configuration $C$ 
is called an \emph{execution from $C$}. A configuration $C$ is said to 
be \emph{reachable} from a configuration $C'$, and denoted $C\in\mathit{Reach}(C')$, 
if there exists a finite execution $\alpha$ applicable to $C'$, such that~$C=C'\alpha$. 
If omitted, the starting configuration is the initial configuration,
and is denoted as~$C\in \mathit{Reach}$. 

Processes that take at least one step of the algorithm are called \emph{participating}.
A process is called \emph{correct} in a given (infinite) execution 
if it takes infinitely many steps in that execution.
Let $\textit{Correct}(\alpha)$ denote the set of 
correct processes in the execution $\alpha$.

\subsection{Comparison-based algorithms}

We assume that the processes are allowed to use their identifiers 
only to compare them with the identifiers of other processes: 
the outputs of the algorithm only depend on the inputs,
the relative order of the identifiers of the participating processes,
and the schedule of their steps.
Formally, we say that an algorithm is \emph{comparison-based}, if, 
for each possible execution $\alpha$, by replacing the identifiers
of participating processes with new ones preserving their relative order, 
we obtain a valid execution of the algorithm. 
Notice that the assumption does not preclude using the identifiers 
in communication primitives, it only ensures that decisions taken in the
algorithm's run are taken only based on the identifiers relative order.

In this model, $m$ MWMR registers can be used to implement a wait-free 
$m$-component \emph{multi-writer atomic-snapshot} memory~\cite{AAD93}.
The memory exports operations $\Update(i,v)$ 
(updating position~$i$ of the memory with value $v$) 
and $\Snapshot()$ (atomically returning the contents of the memory).
In the comparison-based atomic-snapshot implementation, 
easily derived from the original one~\cite{AAD93}, 
$\Update(i,v)$ writes only once, to register $i$, and $\Snapshot()$ is read-only. 
For convenience, in our upper-bound algorithm we are going to use 
atomic snapshots instead of read-write registers.

\subsection{SWMR memory}\label{subsec:safety}
A single-writer multi-reader (SWMR) memory exports two operations: 
$\Write()$ that takes a value as a parameter and 
$\Collect()$ that returns a \emph{multi-set} of values.
It is guaranteed that, in every execution, there exists a reading map $\pi$
that associates each complete Collect operation~$C$, 
returning a multi-set $V=\{v_1,\ldots,v_s\}$, 
with a set of $s$ Write operations $\{w_1,\ldots,w_s\}$ performed,
respectively, by distinct processes $p_1,\dots,p_s$ such that:

\begin{itemize}  
\item The set $\{p_1,\ldots,p_s\}$ contains all processes that completed 
	at least one write operation before the invocation of $C$;  
  
\item For each $i=1,\ldots,s$, $w_i$ is either the last write operation 
	of process $p_i$ preceding the invocation of $C$ 
	or a write operation of $p_i$ concurrent with $C$. 
\end{itemize}

Note that our definition does not guarantee atomicity of SWMR operations. 
Moreover, we do not require that processes are allocated with 
a unique MWMR register that can be used as a single writer register. 
Instead, we simply require that processes are able to simulate the 
use of single writer registers through implementing the SWMR memory.

Intuitively, a collect operation can be seen as a sequence of reads on
\emph{regular} registers~\cite{Lam86}, each associated with a distinct
participating process.
Such a collect object can be easily transformed into a
\emph{single-writer} atomic snapshot abstraction~\cite{AAD93}.

\subsection{Progress conditions}\label{subsec:progress}

In this paper we focus on two families of progress conditions, 
both generalizing the \emph{wait-free} progress condition, 
namely $k$-\emph{lock-freedom} and $k$-\emph{obstruction-freedom}.

An execution $\alpha$ satisfies the property of 
$k$-\emph{lock-freedom}~\cite{BG15} (for $k\in\{1,\dots, n\}$) 
if at least $\min(k,\textit{Correct}(\alpha))$ correct processes 
\emph{make progress} in it, i.e., complete infinitely many high-level operations 
(in our case, Writes and Collects).  
The special case of $n$-lock-freedom is called \mbox{\emph{wait-freedom}}.  
The property of $k$-\emph{obstruction-freedom}~\cite{HLM03,Tau09}
requires that every correct process makes progress, under
the condition that there are at most $k$ correct processes.
(If more than $k$ processes are correct, no progress is guaranteed.) 

In particular, $k$-lock-freedom is a stronger requirement than $k$-obstruction-freedom 
(strictly stronger for~$1\leq k< n$). Indeed, both require that every correct 
process makes progress when there are at most $k$ correct processes, 
but $k$-lock-freedom additionally requires that some progress is made 
even if there are more than $k$ correct processes. 

\section{Upper bound: k-lock-free SWMR memory with n+k-1 registers}
\label{sec:alg}

Consider a \emph{full-information} algorithm in which every process 
alternates atomic snapshots and updates, where each update performed 
by a process incorporates the result of its preceding snapshot.
Every value written to a register will \emph{persist} 
(i.e., will be present in the result of every subsequent snapshot), 
unless there is another process poised to write to that register.
The pigeonhole principle implies that $k$ processes can cover at most
$k$ distinct registers at the same time.
Thus, if, at a given point of a run, a value is present in $n$ registers,
then the value will persist. 
This observation implies a simple $n$-register \emph{lock-free} SWMR implementation
in which a high-level Write operation alternates snapshots and updates
of all registers, one by one in the round-robin fashion, 
until the written high-level value  is present in all $n$ registers.
A high-level Collect operation can simply return the set of the most recent values 
(defined using monotonically growing sequence numbers) returned by a snapshot operation. 

The \emph{wait-free} SWMR memory implementation in~\cite{DFGR15} using $2n-1$ registers 
follows the \mbox{$n$-register} lock-free algorithm but, roughly,
for each participating process, replaces register~$n$ with register $n-1+\textit{pos}$, 
where $\textit{pos}$ is the rank of the process among the currently
observed participants. This way, there is a time after which 
every participating  process has a dedicated register to write, 
and each value it writes will persist.
In particular, every value it writes will be seen by all processes and 
will eventually be propagated to the $n-1$ first registers.

To implement a $k$-lock-free SWMR memory using $n+k-1$ registers,
a process should determine, in a dynamic fashion, 
to which out of the last $k-1$ registers to write.
In our algorithm, by default, a Write operation only uses  the first $n$ registers, 
but if a process observes that its value is absent from some registers in the snapshot
(some of its previous writes have been overwritten by other processes), 
it uses extra registers to propagate its value. 
The number of these extra registers depends on 
how many other processes have been observed making progress. 

\begin{algorithm}
\begin{small}
 \caption{$k$-lock-free SWMR implementation using $n+k-1$ MWMR registers.\label{Alg:k-LockFree}}
\SetKwRepeat{Write}{Write(v):}{End Write}
\SetKwRepeat{Collect}{Collect():}{End Collect}
\SetKwRepeat{Do}{do}{while}
$\mathit{View}:$ \textbf{list of triples of type} $(\mathit{ValueType},\mathit{IdType},\mathbb{N})$, \textbf{initially set to} $\emptyset$\;
$\mathit{opCounter} \in \mathbb{N}$, \textbf{initially set to} $0$\;

\vspace{1em}

\Write{}{
	$\mathit{ActiveProcs} = \{\mathit{id}\}$\;
	$\mathit{View} = \mathit{View}\cup (v,\mathit{id},\mathit{opCounter})$\;\label{Alg:l:kLF_AddOp}
	$\mathit{WritePos} = 0$\;	
	$\mathit{WritePosMax} = n$\;
	\Do{
	$|\{m\in  \{1, \dots, n+k-1\}, (v,\mathit{id},\mathit{opCounter}) \in \mathit{Snap}[m]\}|< n$\label{Alg:l:kLF_ValidateWrite}
	} {\label{Alg:l:kLF_WhileStart}
		$\mathit{Snap} = \mathit{MEM}.\mathit{snapshot}()$\;\label{Alg:l:kLF_TakeSnapshot}
		$\mathit{ActiveProcs} = \mathit{ActiveProcs} \cup \{\mathit{pid}: \exists (\_,\mathit{pid},c)\in \mathit{Snap}, \forall (\_,\mathit{pid},c')\in \mathit{View}, c>c'\}$\;
		$\mathit{View} = \mathit{View}\cup \mathit{Snap}$\;\label{Alg:l:kLF_AddSnapshot}
		$\mathit{Update}(\mathit{MEM}[\mathit{WritePos}],\mathit{View})$\;\label{Alg:l:kLF_WriteValue}
		$\mathit{WritePos} = \mathit{WritePos} + 1 \pmod{\mathit{WritePosMax}}$\;\label{Alg:l:kLF_WritePos}
		$\mathit{WritePosMax} = min(n+|\mathit{ActiveProcs}|-1,n+k-1)$\;\label{Alg:l:kLF_IncrMaxWritePos}
	}\label{Alg:l:kLF_WhileEnd}
	$\mathit{opCounter}=\mathit{opCounter}+1$\;
}
\vspace{1em}

\Collect{}{
	$\mathit{Reads} = \mathit{MEM}.\mathit{snapshot}()$\;
	$ V = \emptyset$\;
	\ForAll{$pid$ \textbf{such that} $(\_,pid,\_)\in \mathit{Reads}$}{
		$V = V\cup\{v\}$ \textbf{with} $v$ \textbf{such that} $(v,\mathit{pid},\max\{c\in \mathbb{N},(\_,\mathit{pid},c)\in \mathit{Reads}\})\in \mathit{Reads}$\;\label{Alg:l:ReadSelect}
	}	
	\textbf{Return} $V$\;
}
\end{small}
\end{algorithm}

\subsection{Overview of the algorithm}
Our $k$-lock-free SWMR implementation, which uses $n+k-1$ base MWMR registers, 
is presented in Algorithm~\ref{Alg:k-LockFree}. 

In a Write operation, the process adds the operation to be performed 
to its local view~(line~\ref{Alg:l:kLF_AddOp}). 
The process then attempts to add its local view, together with the
outcome of a snapshot, to each of the first $\mathit{WritePosMax}$, 
initially $n$, registers (lines~\ref{Alg:l:kLF_WhileStart}--\ref{Alg:l:kLF_WhileEnd}).
At each loop, $\mathit{WritePosMax}$ is set to the smaller value between 
the number of processes observed as concurrently active 
and the number of registers available (line~\ref{Alg:l:kLF_IncrMaxWritePos}).
The writing process continues to do so until its Write operation
value is present in at least $n$ registers (line~\ref{Alg:l:kLF_ValidateWrite}).

In this algorithm, the $k-1$ extra registers are used according to the 
liveness observed by blocked processses. In order to be allowed to 
use the last register, a process must fail to complete its write while 
observing at least $k-1$ other processes completing their own. This 
ensures that when a process access this last register, a $k^{th}$ process
is able to be observed by processes completing operations and thus will
be helped to eventually complete.

The Collect operation is rather straightforward.
It simply takes a snapshot of the memory and, 
for each participating process observed in the memory,
it returns its most recent value (selected using associated sequence numbers, line~\ref{Alg:l:ReadSelect}).

\subsection{Safety} 

At a high level, the safety of Algorithm~\ref{Alg:k-LockFree} relies 
on the following property of register content stability:

\begin{lemma}
Let, at some point of a run of the algorithm, value $(v,id,c)$ be
present in some register~$r$ and such that  
no process is poised to execute an update on $r$ (i.e., no process is between taking the snapshot of $\textit{MEM}$ 
(line~\ref{Alg:l:kLF_TakeSnapshot}) and the update of $r$ 
(line~\ref{Alg:l:kLF_WriteValue})), 
then at all subsequent times $(v,id,c) \in r$, i.e., the value is
present in the set of values stored in $r$.\label{Lem:PersistentuncoveredWrites}
\end{lemma}

\begin{proof}
Suppose that at time $\tau$, a register $R$ contains $(v,id,c)$ 
and no process is poised to execute an update on $R$.
Suppose, by contradiction, that $R$ does not contain it at some time~$\tau'>\tau$. 
Let~$\tau_{min}$,~$\tau_{min}>\tau$, be the smallest time 
such that $(v,id,c)$ is not in $R$. 
Therefore, a write must have been performed on $R$, by some process $q$, 
at time $\tau_{min}$ with a view which does not contain~$(v,id,c)$. 
Such a write can only be performed at line~\ref{Alg:l:kLF_WriteValue}, 
with a view including the last snapshot of $MEM$ performed by $q$ 
at line~\ref{Alg:l:kLF_TakeSnapshot}. 
Process $q$ must have performed this snapshot on $R$ at some~$\tau_R<\tau$ as~$(v,id,c)$ is 
present in $R$ between times $\tau$ and $\tau_{min}$ and as $\tau_R<\tau_{min}$. 
Thus $q$ is poised to write on~$R$ at time $\tau$ --- a contradiction.
\end{proof}

The persistence of the values in a specific uncovered register  
(Lemma~\ref{Lem:PersistentuncoveredWrites}) can be used to show
the persistence of the value of a completed Write operation in $\textit{MEM}$: 

\begin{restatable}{lemma}{persWrite}
If process $p$ returns from a Write operation $(v,id(p),c)$ at time $\tau$, 
then for any time~$\tau'\geq \tau$ there is a register containing $(v,id(p),c)$.
\label{Lem:PersistentWrites}
\end{restatable}

\begin{proof}
Before returning from its Write operation, 
$p$ takes a snapshot of $\textit{MEM}$ 
at some time~$\tau_S$, $\tau_S<\tau$ (line~\ref{Alg:l:kLF_TakeSnapshot}), 
which returns a view of the memory in which at least 
$n$ registers contain the triplet~$(v,id(p),c)$. 
As $p$ is taking a snapshot at time $\tau_S$,  
at most $n-1$ processes can be poised to perform 
an update on some register at time $\tau_S$. 
As a process can be poised to perform an update on at most one register, 
there can be at most $n-1$ distinct registers covered at time~$\tau_S$. 
Therefore, at time $\tau_S$, there is at least one uncovered register
containing $(v,id(p),c)$, let us call it~$r$.
By Lemma~\ref{Lem:PersistentuncoveredWrites}, 
$(v,id(p),c)$ will be present in $r$ at any time $\tau'>\tau_S$, 
and thus, any time $\tau'>\tau$.
\end{proof}

With Lemma~\ref{Lem:PersistentWrites}, we can derive 
the safety of our SWMR memory implementation (Section~\ref{subsec:safety}):

\begin{theorem}
Algorithm~\ref{Alg:k-LockFree} safely implements an SWMR memory.
\end{theorem}

\begin{proof}
It can be easily observed that a triplet $(v,id,c)$ 
corresponds to a unique Write operation of a value $v$, 
performed by the process with identifier $id$. 
Therefore, a Collect operation returns a set of 
values proposed by Write operations from distinct processes,
and thus the map $\pi$ is well-defined.

By Lemma~\ref{Lem:PersistentWrites}, the value $(v,id,c)$ 
corresponding to a Write operation completed at time~$\tau$ 
is present in some register $r$ for any time $\tau'>\tau$. 
Thus, the set of values resulting from any snapshot operation performed 
after time $\tau$ contains $(v,id,c)$.
Hence, for any complete Collect operation $C$, $\pi(C)$ contains   
a value for every process which completed a Write operation
before $C$ was invoked. Also, as each value 
returned by a Collect is the value observed associated
to the greatest sequence number for a given process,
it comes from the last completed Write or
from a concurrent one.
\end{proof}

\subsection{Progress}
We will show, by induction on $k$, 
that Algorithm~\ref{Alg:k-LockFree} satisfies $k$-lock-freedom.
We first show, as in~\cite{DFGR15}, that Write 
operations of Algorithm~\ref{Alg:k-LockFree} are $1$-lock-free:

\begin{lemma}
Write operations in Algorithm~\ref{Alg:k-LockFree} 
satisfy $1$-lock-freedom.\label{Lem:1lock-Free}
\end{lemma}

\begin{proof}
Suppose, by way of contradiction, that Write operations do not satisfy \mbox{$1$-lock-freedom}.
Eventually, all $n$ first registers are infinitely often updated only 
by correct processes unsuccessfully trying to complete a Write operation. 
Thus, eventually each of  the $n$ first registers contain the value
from one of these incomplete Write operations. 
As there are at most $n-1$ covered registers when a snapshot is taken,  
one of these value is eventually permanently present in some 
register~(Lemma~\ref{Lem:PersistentuncoveredWrites}). 
This value is then eventually contained in the local view 
of every correct process, and thus, will eventually be present 
in every update of all the $n$ first registers. 
The correct process with this Write value 
must therefore eventually pass the test 
on line~\ref{Alg:l:kLF_ValidateWrite} and, 
thus, complete its Write operation --- a contradiction. 
\end{proof}

The induction step relies primarily on the helping mechanism.
This mechanism guarantees that a process making progress eventually ensures
that the processes it observes as having a pending operation also make progress 
(the mechanism is similar to the one of the wait-free SWMR memory implementation 
of~\cite{DFGR15}): 

\begin{lemma}
If a process $q$ performing infinitely many operations sees $(v,id(p),c)$, 
and if $p$ is correct, then $p$ eventually completes its~$c^{th}$ Write operation.
\label{Lem:LockFreeHelping}
\end{lemma}

\begin{proof}
By Lemma~\ref{Lem:PersistentWrites}, if process $q$ returns
from a Write operation with value $(v,id(q),c')$ at time $\tau$, then for any time
$\tau'\geq \tau$ there is a register containing $(v,id(q),c')$. But
note that~$(v,id(q),c')$ is written to a register only associated with
$q$'s local view. Thus, as $q$ completes an infinite number of Write
operations, each local view of $q$ will eventually be forever present
in some register, in particular $(v,id(p),c)$. 
Thus $(v,id(p),c)$ is eventually observed in every
snapshot taken by correct processes, and, therefore, included in their
local view. This implies that it will eventually be present in every register written
infinitely often, in particular in the first $n$ registers. 
As $p$ is correct, it eventually sees $(v,id(p),c)$ in $n$ registers
for the test at line~\ref{Alg:l:kLF_ValidateWrite} and, thus, 
completes its corresponding $c^{th}$ Write operation.
\end{proof}

By the base case provided by
Lemma~\ref{Lem:1lock-Free}  and Lemma~\ref{Lem:LockFreeHelping}, we have:

\begin{lemma}
Write operations in Algorithm~\ref{Alg:k-LockFree} 
satisfy $k$-lock-freedom.\label{Lem:Klock-Free}
\end{lemma}

\begin{proof}
  We proceed by induction on $k$, starting with the base case of $k=1$ (Lemma~\ref{Lem:1lock-Free}). 
  Suppose that Write operations satisfy $\ell$-lock-freedom for some  $\ell< k$. 
Consider a run in which at least $\ell+1$ processes are correct, 
but only $\ell$  of them make progress 
(if  such a run doesn't exist, the algorithm satisfies
$(\ell+1)$-lock-freedom).
In this run, at least one correct process is eventually blocked in a Write operation. 
According to Lemma~\ref{Lem:LockFreeHelping}, 
the $\ell$ processes performing infinitely many Write operations 
eventually do not observe new values written by other processes. 
By the algorithm, these processes eventually never write to the last~$k-\ell>0$ registers. 

A correct process that never completes a Write operation 
will execute the while loop 
\mbox{(lines~\ref{Alg:l:kLF_WhileStart}--\ref{Alg:l:kLF_WhileEnd})}
infinitely many times, 
and thus, will infinitely often take a snapshot and 
update its local view~(line~\ref{Alg:l:kLF_TakeSnapshot}). 
In particular, it will eventually observe a new Write operation 
performed by each of the $\ell$ processes completing 
infinitely many Write operations. 
It will then eventually include at least $\ell+1$ processes 
in its set of active processes (i.e., the $\ell$ processes performing 
infinitely many Write operations and~itself). 
It will therefore eventually write to the $(n+\ell)^{th}$ register 
infinitely often.  
In the considered run, this register is written infinitely often only by correct processes 
which do not complete new Write operations.
The value from at least one of such process will then be observed by the $\ell$ processes 
making progress. 
By Lemma~\ref{Lem:LockFreeHelping}, 
this process 
will eventually complete its Write operation --- a contradiction.
\end{proof}

Collect operations in Algorithm~\ref{Alg:k-LockFree} clearly satisfy wait-freedom 
as there are no loops and MWMR snapshot operations are wait-free. 
Thus Lemma~\ref{Lem:Klock-Free} and the wait-freedom 
of Collect operations imply~that:

\begin{theorem}
Algorithm~\ref{Alg:k-LockFree} is a  $k$-lock-free 
implementation of an SWMR memory for $n$ processes using $n+k-1$ MWMR registers.
\end{theorem}

\section{Lower bound:  impossibility of 2-obstruction-free SWMR memory implementations with n MWMR registers}
\label{sec:lb}

The algorithm in Section~\ref{sec:alg} gives an upper bound of $n+k-1$
on the number of MWMR registers required to implement an SWMR memory
satisfying the $k$-lock-free progress condition in the comparison-based model.
In this section, we present a lower bound on the number of MWMR registers required 
in order to provide a $2$-obstruction-free, and hence also a $2$-lock-free, SWMR memory
implementation. 

\subsection{Overview of the lower bound}

Our proof relies on the concepts of {\it covering} and {\it indistinguishability}.

A register is \emph{covered} at a given point of a run if there is at least
one process poised to write to it (we say that the process covers the
register).  Hence, a covered register cannot be used to ensure
persistence of written data: by awakening the covering
process, the adversarial scheduler can overwrite it.
This property alone can be used to show that $n$ registers are required for 
an \mbox{obstruction-free} (and hence also for a $1$-lock-free) SWMR memory
implementation~\cite{BL83}, but not to obtain a lower bound of 
\emph{more} than $n$ shared resources as there is always one which 
remains~uncovered. 

Indistinguishability captures bounds on the knowledge that a process has 
of the rest of the system.
Two system states are \emph{indistinguishable} for a process if it has the same
local state in both states and if the shared memory includes the same content.
Thus, in an SWMR memory implementation, a Write operation can safely terminate
only if, in all indistinguishable states, its value is present in a
register that is not covered (by a process unaware of that value).

In our proof, we work with a composed notion of 
covering and indistinguishability.
The idea is to show that there is
a large set of reachable system states, 
indistinguishable to a given process~$p$,
in which different \emph{sets of registers} are covered.
Intuitively, if a set of registers is covered in one of these indistinguishable  
states,~$p$ must necessarily write to a register outside of this 
set in order to 
complete a new Write operation.
Hence, if such indistinguishable states exist for \emph{all} register
subsets, then $p$ must write its value to \emph{all} registers.
To perform  infinitely many high-level Write operations,~$p$ must
then write infinitely often to all available registers.
But then any other process $p'$ taking steps can be masked by the execution of $p$ (i.e., 
any write $p'$ makes to a MWMR register can be scheduled to be overwritten by $p$). 
This way we establish that no $2$-obstruction free implementation
exists, as it requires that 
at least two processes
must be able to make
progress concurrently.

\subsection{Preliminaries} 

Assume, by contradiction, that there exists 
a $2$-obstruction-free SWMR implementation using only~$n$ registers. 
To establish a contradiction, we consider a set of runs by a fixed set
$\Pi$ of $n$ processes in which every
process performs infinitely many Write operations with monotonically
increasing arguments.
Let $\mathcal{R}$ denote the set of $n$ available registers.

\paragraph*{Indistinguishability}

A configuration $C$ is said to be \emph{indistinguishable} from a configuration $C'$ 
for a set of processes~$P$, if the content of all registers and the states
of all processes in $P$ are identical in~$C$ and $C'$.
Given a set of configurations $\cal D$, let $I({\cal D},P)$ denote that 
any two configurations from~$\D$ are indistinguishable for $P$. 

We say that an execution is $P$-only, for a set of processes $P$, 
if it consists only of steps by processes in $P$.
We say that a set of processes $P$ is \emph{hidden} in an execution
$\alpha$ if all writes in~$\alpha$ performed by processes in $P$
are overwritten by some processes not in $P$, 
without any read performed by processes not from $P$ in between.
Given a sequence of steps $\alpha$ and a set of processes~$P$, 
let $\alpha|_P$ be the sub-sequence of $\alpha$ 
containing only the steps from processes in $P$.
Let us denote as $\D\alpha$ the set of all configurations 
reached by applying $\alpha$ to all configurations in~$\D$
(note that $\alpha$ must be applicable to all configurations in~$\D$).

\begin{observation}
  If a $P$-only execution $\alpha$ is applicable to a configuration $C$ from a set of 
  configurations $\cal D$ indistinguishable for $P$, i.e., $C\in{\cal D}$ and $I({\cal D},P)$,
  then $\alpha$ is applicable to any configuration~$C'\in {\cal D}$,  
  and it maintains the indistinguishability of configurations for $P$, i.e.,
	 $I(\D\alpha,P)$.
	 \label{Obs:P-only}
\end{observation}

A similar observation can be made concerning hidden executions:

\begin{observation}
Given an execution $\alpha$ applicable to $C$, with $C$ 
from a set of configurations~$\cal D$ indistinguishable for P. 
If processes in $\Pi\setminus P$ are hidden in $\alpha$, 
then $\alpha|_P$ is applicable to any~$C'\in {\cal D}$, 
and $I((\D\alpha|_P)\cup\{C\alpha\},P)$.
\label{Obs:hidden}
\end{observation}

\paragraph*{Coverings and confusion} 

We say that a set of processes $P$ \emph{covers}
a set of registers $R$ in some configuration $C$,
if for each register $r\in R$,
there is a process $p\in P$ such that the next step 
of $p$ in~$C$ is a write on $r$  (the predicate is denoted $\mathit{Cover}(R,P,C)$).

Our lower bound result relies
on a concept that
we call \emph{confusion}.
We say that a set of processes $P$ are \emph{confused} on a set of registers $S$ 
in a set of reachable configurations $\cal D$, 
denoted~$\mathit{Confused}(P,S,\cal D)$, if and only if:
\begin{enumerate}
\item $I(\mathcal{D},P)$.
\item $|S|+|P|=n+1$.
\item For any process $p\in\Pi\setminus P$, 
there exist two registers $r_p,r_p'\in S$ such that, for any
configuration~$D\in\cal D$, there exists $D'\in \cal D$, such that
$p$ covers $r_p$ in $D$ and $r_p'$ in $D'$, or vice versa, 
and~$D$ and $D'$ are indistinguishable to all other processes:
\[
	\forall p\in \Pi\setminus P,\exists r_p,r_p' \in S,
	\forall D\in{\cal D},\exists r \in \{r_p,r_p'\}:
	\]\[
	 \mathit{Cover}(\{r\},\{p\},D)
	\wedge(\exists D'\in {\cal D}, I(\{D,D'\},\Pi\setminus\{p\})\wedge \mathit{Cover}(\{r_p,r_p'\}\setminus\{r\},\{p\},D')){}.
\]
\item For any strict subset $R$ of $S$, there exists $D\in\mathcal{D}$ such that 
	$R$ is covered by $\Pi\setminus P$ in $D$:
\[
	\forall R\subsetneq S, \exists D\in{\cal D}: \mathit{Cover}(R,\Pi\setminus P, D){}.
\]
\end{enumerate}

Intuitively, processes in $P$ are confused on $S$ in $\cal D$,
if $\cal D$ is a set of indistinguishable configurations for $P$,
such that any \emph{strict} subset of $S$ is covered by $\Pi\setminus P$
in some configuration of $\cal D$ (Conditions~$1$ and $4$).
We require that as much processes are confused as possible (Condition $2$).
Additionally, the property must hold for a set of configurations $\D$
in which processes not in $P$ may cover only one out of 2 given registers, 
and may be cover them independently of other processes states in~$\D$~(Condition~$3$).

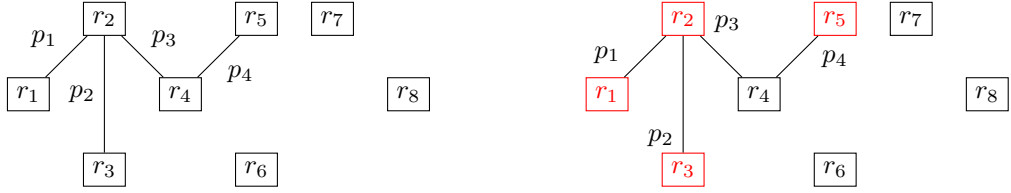
\begin{figure}
\begin{center}
\begin{minipage}{0.49\textwidth}
\begin{center}
\begin{tikzpicture}
  \node[draw] (A) at (1,2) {$r_1$};
  \node[draw] (B) at (2,3) {$r_2$};
  \node[draw] (C) at (2,1) {$r_3$};
  \node[draw] (D) at (3,2) {$r_4$};
  \node[draw] (E) at (4,3) {$r_5$};
  \node[draw] (F) at (4,1) {$r_6$};
  \node[draw] (G) at (5,3) {$r_7$};
  \node[draw] (H) at (6,2) {$r_8$};
  \draw (A) -- (B) node[midway,above left] {$p_1$};
  \draw (B) -- (C) node[midway,left] {$p_2$};
  \draw (B) -- (D) node[midway,above right] {$p_3$};
  \draw (D) -- (E) node[midway,below right] {$p_4$};
\end{tikzpicture}
\end{center}
\end{minipage}
\hfill
\begin{minipage}{0.49\textwidth}
\begin{center}
\begin{tikzpicture}
  \node[draw,red] (A) at (1,2) {$r_1$};
  \node[draw,red] (B) at (2,3) {$r_2$};
  \node[draw,red] (C) at (2,1) {$r_3$};
  \node[draw] (D) at (3,2) {$r_4$};
  \node[draw,red] (E) at (4,3) {$r_5$};
  \node[draw] (F) at (4,1) {$r_6$};
  \node[draw] (G) at (5,3) {$r_7$};
  \node[draw] (H) at (6,2) {$r_8$};
  \draw (A) -- (B) node[very near start,above left] {$p_1$};
  \draw (B) -- (C) node[very near end,left] {$p_2$};
  \draw (B) -- (D) node[very near start,above right] {$p_3$};
  \draw (D) -- (E) node[very near end,below right] {$p_4$};
\end{tikzpicture}
\end{center}
\end{minipage}
\caption{Processes $\{p_5,p_6,p_7,p_8\}$ are confused on
 registers $\{r_1,r_2,r_3,r_4,r_5\}$; an example of a possible
  covering is given on the right.\label{Fig:confusion}}
\end{center}
\end{figure}

In Figure~\ref{Fig:confusion}, we give an example of a confusing
set of configuration $\cal D$ for $8$ processes and~$8$ registers. 
Processes $\{p_5,p_6,p_7,p_8\}$ are confused on registers  
$\{r_1,r_2,r_3,r_4,r_5\}$. Registers are 
represented as nodes, and pairs of registers that a process might be 
covering are represented as edges.
The set of indistinguishable configurations $\cal D$ for $\{p_5,p_6,p_7,p_8\}$ 
are defined via composition of states for $p_1$, $p_2$, $p_3$ and
$p_4$ in which they, respectively, cover registers in $\{r_1,r_2\}$, $\{r_2,r_3\}$, 
$\{r_2,r_4\}$ and~$\{r_4,r_5\}$. An example of a covering of $\{r_1,r_2,r_3,r_5\}$
for some particular execution is presented on the right side of Figure~\ref{Fig:confusion}.

First, we are going to provide an alternative property for Condition
$4$ of the definition of $\mathit{Confused}(P,S,\D)$. 
The idea is that, given $(P,S,\D)$ satisfying Conditions $1$, $2$ and $3$, 
Condition $4$ is satisfied if and only if the graph induced by the sets of 
registers that may be covered by processes in $\Pi\setminus P$ 
(as represented in Figure~\ref{Fig:confusion}) forms a connected component over~$S$. 
More formally, that Condition $4$ is satisfied if and only if, for any partition of $S$ 
into two non-empty subsets $S_1$ and $S_2$, there is a process in $\Pi\setminus P$ 
for which the set of two registers it may be covering in $\cal D$ intersects 
with both $S_1$ and $S_2$:

\begin{lemma}
$\forall P\subseteq \Pi, \forall S\subseteq \mathcal{R}, 
\forall {\cal D}\subseteq \mathit{Reach}$ satisfying Conditions $1$, $2$ and $3$ of the 
confusion definition, we have $\forall R\subsetneq S, \exists D\in{\cal D}: \mathit{Cover}(R,\Pi\setminus P, D)$ if and only if:
\[\forall S_1,S_2\subseteq S, (S_1\neq\emptyset\wedge S_2\neq\emptyset \wedge 
S_1\cup S_2= S\wedge S_1\cap S_2 =\emptyset): 
\]\[
\exists r_1\in S_1, r_2\in S_2, p\in \Pi\setminus P, D_1,D_2\in {\cal D}:(\mathit{Cover}(\{r_1\},\{p\},D_1)\wedge\mathit{Cover}\left((\{r_2\},\{p\},D_2)\right){}.
\]
\label{lem:interCharac}
\end{lemma}

\begin{proof}
Let us fix some $P\subseteq \Pi$, $S\subseteq \mathcal{R}$, 
and ${\cal D}\subseteq \mathit{Reach}$ satisfying 
Conditions $1$, $2$ and $3$ of the confusion definition.

First, let us assume that Condition $4$ is also satisfied 
and consider any partition of $S$ into non-empty subsets $S_1$ and $S_2$
(i.e., $S_1\neq\emptyset$, $S_2\neq\emptyset$, $S_1\cap S_2=\emptyset$ and $S_1\cup S_2=S$).
Assume now that there does not exist any process $p\in\Pi\setminus P$ such that $p$ 
might be covering a register from~$S_1$ or a register from $S_2$ in $\cal D$. 
This implies that processes in $\Pi\setminus P$ can be partitioned 
into two subsets~$Q_1$ and $Q_2$ (with $Q_1\cap Q_2=\emptyset$ and $Q_1\cup Q_2=\Pi\setminus P$)
such that processes in $Q_1$, respectively $Q_2$, may cover registers 
from $S_1$, respectively $S_2$, in $\cal D$.
By construction of the partitions, we have $|Q_1|+|Q_2|= |\Pi\setminus P|$ 
and $|S_1|+|S_2|=|S|$. Using the fact that Condition $2$ is satisfied by $P$ 
and $S$ we obtain from $|S|+|P|= n+1$ that $|S_1|+|S_2|+(n-(|Q_1|+|Q_2|))=n+1$, 
and thus, that $|S_1|+|S_2|=|Q_1|+|Q_2|+1$. 
This implies that either $|Q_1|<|S_1|$ or $|Q_2|<|S_2|$, w.l.o.g., let~$|Q_1|<|R_1|$. 
Now consider $r\in S_2$, $S\setminus\{r\}$ is a
strict subset of $S$, and therefore Condition $4$ implies that there exists 
$D\in\mathcal{D}$ such that $\mathit{Cover}(S\setminus\{r\},\Pi\setminus P, D)$.
As registers in $S_1$ can only be covered by processes from $Q_1$, then we have
$\mathit{Cover}(S_1,Q_1, D)$. Recall that, by the pigeonhole principle, 
a set of processes cannot cover more registers than processes it contains. 
But $|Q_1|<|R_1|$ --- a contradiction.

Now let us assume that given any partition of $S$ into 
non-empty subsets $S_1$ and $S_2$, there exists a process $p\in \Pi\setminus P$
such that $p$ might be covering a register in $S_1$ or a register in $S_2$ in $\cal D$.
Let us show that any strict subset $R$ of $S$ is covered in some 
configuration from $\mathcal{D}$ and, hence, that Condition $4$ is satisfied.
This is done by inductevely restricting the set of configurations from $\D$, 
by selecting the maximal subset in which some process from $\Pi\setminus P$ 
may cover only one register. The idea is to select a process
which may cover only one not-yet covered register in $R$, and to select
the subset in which this process covers this register.

Let $S_0$ be a non-empty subset of $S$. Let $p_0$ be a process from $\Pi\setminus P$ which 
might be covering a register $r_0$ in $S_0$ or a register $r_0'$ in $S\setminus S_0$ in $\cal D$. 
Let us assume that such a process exists and consider $\mathcal{D}_0$ to be the 
subset of $\mathcal{D}$ including all configurations in which~$p_0$ 
is covering~$r_0'$. Now let $S_1=S_0\cup\{r_0'\}$ and repeat this procedure
using $S_1$ to select some $p_1$ and compute~${\cal D}_1$, etc... 
As long as a process can be selected satisfying the condition, 
the sets~$S_i$ keep increasing with $i$. 
Consider the round $j$ at which the procedure fails to find such a process. 
This implies that there is no process which might be covering a register 
from either $S_j$ or $S\setminus S_j$ in ${\cal D}_{j-1}$. Note that by 
construction ${\cal D}_{j-1}$ is a non-empty subset of ${\cal D}$.

If $S_j\neq S$, then $S_j$ and $S\setminus S_j$ forms a partition 
of $S$ into two non-empty subsets. Thus, by assumption, 
there exists a process $q$ which might be covering a register $r_q$
in $S_j$ or a register $r_q'$ in~$S\setminus S_j$ in $\cal D$. 
Consider some configuration $D\in {\cal D}_{j-1}$.
According to Condition $3$ of the confusion definition, 
as $D\in {\cal D}$, $q$ is covering either $r_q$ or $r_q'$ in $D$, 
and there exists a configuration $D'$ in which $q$ is covering the
other register in $\{r_q,r_q'\}$, relatively to $D$, and such that $D$ 
and $D'$ are indistinguishable to all other processes. 
As $I(\{D,D'\},\Pi\setminus\{q\})$, if $D$ was kept in some 
restriction of~$\D_{i-1}$ towards $\D_i$, then $D'$ was also kept 
unless $q$ was the corresponding selected process $p_i$.
But if $q$ was selected in an earlier iteration, 
both $r_q$ and $r_q'$ would be included in $S_j$. 
Thus $D$ and $D'$ belong to ${\cal D}_{j-1}$ and therefore $q$ is a 
valid selection for $p_j$. This contradiction implies that therefore~$S_j=S$.

By construction, all registers in $S_j\setminus S_0$ are covered in all 
configurations in $\mathcal{D}_{j-1}$. As this is true for any 
non-empty $S_0$ and as $S_j=S$, any strict subset of $S$ is covered 
in some configuration of $\mathcal{D}$ and therefore Condition $4$ is satisfied.
\end{proof}

Lemma~\ref{lem:interCharac} can be used to show that, 
given any confusion for some $P$ distinct from $\Pi$, $S$ and $\D$,
we can identify a process $p\in \Pi\setminus P$ and a register~$r\in S$ 
such that $P\cup\{p\}$ is confused on $S\setminus\{r\}$ for a subset 
$\cal D'$ of $\cal D$ which includes any given $C\in\D$:

\begin{lemma} 
Given $P\subsetneq \Pi$, $S\subseteq \mathcal{R}$, ${\cal D}\subseteq \mathit{Reach}$: 
$\mathit{Confused}(P,S,{\cal D}) \implies$\\
$\exists p\in\Pi\setminus P,\exists r\in S,\forall C\in\D,\exists \D'\subseteq \D:
(C\in\D')\wedge\mathit{Confused}(P\cup\{p\},S\setminus\{r\},\D')$.\label{lem:reducedConfusion}
\end{lemma}

\begin{proof}
According to Condition~$3$, a process may be covering 
exactly two registers from $S$ in~$\D$, thus, the sum over $S$ 
of how many distinct processes may cover each register equals to $2|\Pi\setminus P|=2(n-|P|)$. 
Note that any register $r\in S$ may be covered by at least by one process in 
$\Pi\setminus P$ in $\cal D$ as any strict subset of $S$ may be covered (Condition~$4$).  
Therefore, there exists a register $r_c\in S$ which can be covered 
by a single process $p_c\in\Pi\setminus P$ in $\cal D$. 
Indeed, if all registers in $S$ might be covered by two distinct processes, 
then, the sum over $S$ of how many distinct processes may cover each register
(equal to $2(n-|P|)$), would be greater than or equal to $2|S|$,
or $n-|P|<|S|$ as~$|P|+|S|=n+1$~(Condition~$2$).

Let $\D_c$ be the subset of $\D$ which includes all configurations in~$\D$ 
that are indistinguishable to~$p_c$ from any configuration $C\in \D$.
Let us show that $Confused(P\cup\{p_c\},S\setminus\{r_c\},{\cal D}_c)$. 
Condition~$1$ holds as by construction all configurations in ${\cal D}_c$ are 
indistinguishable to $p_c$ and as they are indistinguishable to all processes in $P$, 
since ${\cal D}_c$ is a subset of $\cal D$.
It is immediate, as we remove a register from $S$ and add a process to $P$, that Condition $2$ holds.

Now consider any process $p\in\Pi\setminus(P\cup \{p_c\})$ and any configuration $D\in \D_c$. 
As $p\in \Pi\setminus P$ and $D\in\D$, Condition $3$ of the confusion definition implies that 
that there exists $D'\in \D$, $I(\{D,D'\},\Pi\setminus\{p\})$, such that $p$ covers  
$r_p$ and $r_p'$ in $D$ and $D'$ respectively (or vice-versa). 
Since $I(\{D,D'\},\Pi\setminus\{p\})$, $D'\in \D_c$, and since $p_c$ is the only process 
which may cover $r_c$ in~$\D$, $r_p$ and $r_p'$ belong to $S\setminus\{r_c\}$. 
Thus Condition $3$ is verified for $P\cup\{p\}$, $S\setminus\{r_c\}$ and $\D_c$.

Lastly, let us consider some partition of $S\setminus\{r_c\}$ 
into two non-empty subsets $S_1$ and~$S_2$. 
Both ($S_1\cup\{r_c\}$,$S_2$) and ($S_1$,$S_2\cup\{r_c\}$) form a partition
of $S$ in two non-empty subsets. Thus, as~$\mathit{Confused}(P,S,{\cal D})$, we can apply
Lemma~\ref{lem:interCharac} and obtain that $\exists p_1,p_2\in \Pi\setminus P$ 
such that~$p_1$, respectively $p_2$, might cover registers from either $S_1\cup\{r_c\}$ or $S_2$, 
respectively either $S_1$ or $S_2\cup\{r_c\}$, in $\cal D$. 
It follows that $p_c$ cannot be both $p_1$ and $p_2$ as $p_c$ might cover only 
two registers in $\cal D$, one of which is $r_c$. Thus, depending 
whether the other register belongs to $S_1$ or $S_2$, $p_1$ or $p_2$ is disctinct from $p_c$. 
W.l.o.g, assume that $p_1\neq p_c$. As $p_c$ is the only process which may be covering 
$r_c$, this implies that $p_1$ might be covering a register from either $S_1$ or~$S_2$.
Furthermore, since Conditions $1$, $2$ and $3$ applies to 
$P\cup\{p\}$, $S\setminus\{r_c\}$ and $\D_c$, we can apply Lemma~\ref{lem:interCharac}
to obtain that Condition $4$ is also verified.
\end{proof}

We now show that the characterization can be used 
to increase the number of registers that processes are 
confused on, by decreasing the number of confused processes:

\begin{lemma}
Let $P\subsetneq \Pi$, $S\subseteq\mathcal{R}$ and $\D\subseteq \mathit{Reach}$ 
such that $\mathit{Confused}(P,S,\D)$. 

\noindent Given $C\in\D$, if $\exists p\in P, r_1\in S, r_2\in \mathcal{R}\setminus S$ and 
if there exist $P$-only executions $\alpha_1$ and~$\alpha_2$ which are
applicable to $C$, and such that $I(\{C\alpha_1,C\alpha_2\},\Pi\setminus\{p\})$, 
$Cover(\{r_1\},\{p\},C\alpha_1)$ and $Cover(\{r_2\},\{p\},C\alpha_2)$, 
then we have $\mathit{Confused}(P\setminus \{p\},S\cup\{r_2\},(\D\alpha_1)\cup(\D\alpha_2))$.
\label{lem:extendConfusion}
\end{lemma}

\begin{proof}
Following Observation~\ref{Obs:P-only}, as $\alpha_1$ and $\alpha_2$ are $P$-only,
and as $\D$ satisfies Condition~$1$ for $P$, $(\D\alpha_1)\cup(\D\alpha_2)$ satisfies Condition $1$ for $P\setminus\{p\}$. 
Condition $2$ trivially holds for $P\setminus\{p\}$ and $S\cup\{r\}$ as it holds 
for $P$ and $S$ and we remove a process from $P$ and add a register to $S$.

Condition $3$ is satisfied for all processes in $\Pi\setminus P$ and configurations in 
$(\D\alpha_1)\cup(\D\alpha_2)$ as~$\alpha_1$ and~$\alpha_2$ are $P$-only, 
and as $\D$ satisfies Condition~$3$ for any process in $\Pi\setminus P$. 
Moreover, as configurations in~$\D$ are indistinguishable to $p\in P$, 
$p$ may only cover $r_1$ if $D\in\D\alpha_1$ and cover $r_2$ if $D\in\D\alpha_2$. 
But as given any $D\in \D$ we have $I(\{D\alpha_1,D\alpha_2\},\Pi\setminus\{p\})$, 
Condition $3$ is also satisfied for $p$.

Since Conditions $1$, $2$ and $3$ are satisfied, we can apply 
Lemma~\ref{lem:interCharac} and obtain that
$\mathit{Confused}(P\setminus \{p\},S\cup\{r_2\},(\D\alpha_1)\cup(\D\alpha_2))$.
Indeed, a partition of two non-empty subsets of~$S\cup\{r_2\}$
can be reduced, unless the partition is $(S,\{r_2\})$, 
to a partition of two non-empty subsets of $S$. In this case, 
Lemma~\ref{lem:interCharac} can be applied for $P$,$S$ and $\D$, 
which provides us, for any partition of $S$, with a process that may cover in $\D$
a register from either set of the partition. As $\alpha_1$ is $P$-only, 
it still holds for $\D\alpha_1$. 
For the partition $(S,\{r_2\})$, $p$ may cover either $r_1\in S$ or $r_2$ in 
$(\D\alpha_1)\cup(\D\alpha_2)$. 
\end{proof}

\subsection{The lower bound}

To establish our lower bound, we show that there
is a set of reachable configuration $\D$ in which 
there is a process confused on all $n$ registers.
Intuitively, we proceed by induction on the number of ``confusing'' registers.  
For the base case, we show that the initial
configuration can lead to a confusion of all but one process on \emph{two registers}:

\begin{lemma}\label{lem:initConfusion}
$\exists \D\in\mathit{Reach},\exists p\in \Pi,\exists S\subseteq \mathcal{R}: \mathit{Confused}(\Pi\setminus\{p\},S,\D){}.$
\end{lemma}

\begin{proof}
Consider any two processes $p_1$ and $p_2$.
Since the algorithm is comparison-based, the first write the two
processes perform in a solo execution is on the same register, let us call it~$r$.
Let $p_1$ execute solo until it is about to write to $r$ and then
do the same with $p_2$, let $C$ be the resulting configuration.
Consider the execution $\alpha$ from $C$ in which $p_1$ executes 
until it is poised to write to a register $r'\neq r$ and
then $p_2$ executes its pending write on $r$.
This execution is valid as $p_1$ must eventually write to an uncovered register.

We obtain $\mathit{Confused}(\Pi\setminus\{p_1\},\{r,r'\},\{C\alpha,C\alpha|_{\{p_2\}}\})$.
Indeed, as $p_1$ is hidden in $\alpha$, following Observation~\ref{Obs:hidden}, 
we have $I(\{C\alpha,C\alpha|_{\{p_2\}}\},\Pi\setminus\{p_1\})$ (Condition $1$). 
We have $|\Pi\setminus\{p_1\}|+|\{r,r'\}|=n+1$ (Condition $2$). 
As $p_1$ covers $r'$ in $C\alpha$ and $p_2$ covers $r$ in $C\alpha|_{\{p_2\}}$, 
we have Condition $4$. Condition $3$ directly follows 
from Conditions $1$ and $4$ in this setting.
\end{proof}

We now prove our inductive step. Given a set of configurations in
which a set of processes,~$P\neq\Pi$, is confused on a set of registers, 
$S\neq\R$, we can obtain a set of configurations in which a set $P'$ of 
processes are confused on a set $S'$ of strictly more than $|S|$ registers:

\begin{lemma}\label{lem:inductiveConfusion}
$\exists {\cal D}\subseteq \mathit{Reach}, P\subsetneq \Pi, 
S\subsetneq \mathcal{R}: \mathit{Confused}(P,S,{\cal D})$\\
$\implies \exists {\cal D'}\subseteq \mathit{Reach}, P'\subseteq
\Pi, S'\subseteq \mathcal{R}, S\subsetneq S':
\mathit{Confused}(P',S',{\cal D'})$.
\end{lemma}

\begin{proof}
Given $\mathit{Confused}(P,S,\D)$, consider $C\in\D$ such that exactly
$|S|-1$ registers in $S$ are covered by processes in $\Pi\setminus P$.
Then we can reach a configuration  
in which all registers not in~$S$ are covered 
by processes in $P$.
Indeed, when executed solo starting from $C$, 
a process must eventually write to a register that is not covered 
in $C$.
Thus, it must eventually write either to a register in $\mathcal{R}\setminus S$
or to the uncovered register in $S$.
Recall that, as $|S|+|P|=n+1$, we have $|\R\setminus S|=|P|-1$.
Thus, by concatenating solo executions of processes in $P$ until 
they are poised to write to uncovered registers,
we reach a configuration $C\alpha$ in which \emph{all} registers are
covered. Let $p$ be the process in $P$ covering a register from $S$ in $C\alpha$. 
Note that, as $\alpha$ is $P$-only, we have $\mathit{Confused}(P,S,\D\alpha)$.
Thus:
\[
\mathit{Cover}(\mathcal{R}\setminus S,P\setminus\{p\},C\alpha)\wedge\mathit{Confused}(P,S,\D\alpha){}.
\]

Now from this set of configurations, we are going to build a new one in which 
$P$ is confused on \emph{two} distinct sets of registers.
By Lemma~\ref{lem:reducedConfusion}, there exist
$p_c\in\Pi\setminus P$ and $r\in S$ such that for any $C'\in\D$ we have
$\mathit{Confused}(P\cup\{p_c\},S\setminus\{r\},\D')$ with $C'\in\D'$.
Let us select $C'\in\D$ to be a configuration in which $p_c$ covers $r_c\in S\setminus\{r\}$ 
(Since we have $\mathit{Confused}(P,S,\D)$, $p_c\in P$ may cover two registers from $S$ in $\D$ 
and so at most one can be $r$).

If $p$ is executed solo from $C'\alpha$, it must write infinitely often to \emph{all} registers in $S$
to ensure that it writes to an uncovered register. 
Hence, in a $\{p,p_c\}$-only execution from $C'\alpha$, $p_c$ can be hidden for arbitrarly many steps 
as long as $p_c$ does not write to a register outside of $S$. 
But, as the algorithm satisfies $2$-obstruction-freedom, $p_c$ \emph{must} eventually 
write to a register outside of $S$ in such an execution. 
Consider the $\{p,p_c\}$-only execution $\beta$ from $C'\alpha$ in which~$p_c$ is hidden 
and such that $p_c$ executes until it is poised to write to some register $r'\in\mathcal{R}\setminus S$.
Thus, we get two configurations $C'\alpha\beta$ and
$C'\alpha\beta|_{\{p\}}$, indistinguishable to all processes but $p_c$, in which $p_c$ covers, respectively, $r'\in \mathcal{R}\setminus S$ and $r_c\in S$.
Thus, the conditions of Lemma~\ref{lem:extendConfusion} hold for $\D'$,
$p_c$, $\alpha\beta$ and~$\alpha\beta|_{\{p\}}$ and so we obtain     
$\mathit{Confused}(P,(S\cup\{r'\})\setminus\{r\}, (\D'\alpha\beta)\cup(\D'\alpha\beta|_{\{p\}}))$.
As $\beta$ is $\{p,p_c\}$-only and $p_c$ is hidden in it, we have:
\[\mathit{Cover}(\mathcal{R}\setminus S,P\setminus\{p\},C\alpha\beta)\wedge
\mathit{Confused}(P,S,\D\alpha\beta|_{\{p\}})\wedge\]\[
\mathit{Confused}(P,S\cup\{r'\}\setminus\{r\},(\D'\alpha\beta)\cup(\D'\alpha\beta|_{\{p\}})){}.
\]

Moreover, all configurations in the formula above are
indistinguishable to processes in $P$, since~$\D'\subseteq\D$, $I(\D,P)$, $\alpha\beta$ is
$P\cup\{p_c\}$-only and $p_c$ is hidden in it (Observation~\ref{Obs:hidden}).     

Let $p'$ be the process from $P$ that covers $r'$ in $C\alpha\beta$.
According to $p$ or $p'$, every proper subset of $S$ or $S\cup\{r'\}\setminus\{r\}$ 
may be covered in the current configuration by $\Pi\setminus (P\cup\{p,p\})$ 
and all other registers covered by $P\setminus\{p,p'\}$.
Thus, from $C\alpha\beta$, to complete a Write operation, $p$ or $p'$ must write
to \emph{all} registers in one of the sets $S$, $S\cup\{r'\}\setminus\{r\}$ or~$\{r,r'\}$.

Consider any $\{p,p'\}$-only extension of $C\alpha\beta$. 
If one of $\{p,p'\}$ covers a register in $S\setminus\{r\}$, $r$ or~$r'$, 
then the other process, in any solo extension, must write respectively 
to all registers in $\{r,r'\}$, $S$ or~$(S\cup\{r'\})\setminus\{r\}$.
In particular, since $p'$ covers $r'$ in~$C\alpha\beta$, $p$ running solo 
from $C\alpha\beta$ must eventually cover a register in $S\setminus\{r\}$ 
(Note that $S\setminus\{r\}\neq\emptyset$, since $|P|<n$ and $|P|+|S|=n+1$). 
Then $p'$ executing solo afterwards must write to $r$ and $r'$.
Let us stop $p'$ when it covers a register~$r''\neq r$ for the last time  before
writing  to $r$. Let $\gamma$ be the resulting execution, and
$E=C\alpha\beta\gamma$ be the resulting~configuration.

Let $\E$ and $\E'$ denote the sets of configurations indistinguishable from $E$
to $P$ defined as $\D\alpha\beta|_{\{p\}}\gamma$ and $(\D'\alpha\beta\gamma)\cup(\D'\alpha\beta|_{\{p\}}\gamma)$ respectively. Note that as $\gamma$ is $P$-only, we still have 
$\mathit{Confused}(P,S,\E)$ and $\mathit{Confused}(P,S\cup\{r'\}\setminus\{r\},\E')$.

Now the following two cases are possible:
\begin{enumerate}
\item $r''\not\in S\cup\{r'\}$: In this case, 
we let $p$ continue until it is poised to write on $r$, 
and then, we let the process from $P\setminus\{p,p'\}$ which 
covers $r''$ to proceed to its pending write on $r''$. 
Let~$\delta$ be this $P$-only execution from $E$ in which $p'$ is hidden. 
As $p'$ covers $r\in S$ in $E\delta$ and~$r''\in\mathcal{R}\setminus S$ 
in $E\delta|_{P\setminus\{p'\}}$, as 
$I(\{E\delta,E\delta|_{P\setminus\{p'\}}\},\Pi\setminus\{p'\})$,
and as $\mathit{Confused}(P,S,\E)$, we can apply 
Lemma~\ref{lem:extendConfusion} and obtain 
$\mathit{Confused}(P\setminus\{p'\},S\cup\{r''\},(\E\delta)\cup(\E\delta|_{P\setminus\{p'\}}))$.

\item $r''\in S\cup\{r'\}$, and so $r''\in (S\cup\{r'\})\setminus\{r\}$: 
  Then we have the following sub-cases:
\begin{itemize}
\item Some step performed by $p$ in its solo execution from $E$ makes
  $p'$ to choose a register other than $r$ to perform its next write
  in its solo extension. Clearly, this step of $p$ is a write.
  From the configuration in which $p$ is poised to execute
  this ``critical'' write, let~$p'$ run solo until it is poised to write
  to $r$ and then let $p$ complete its pending write.
  Let~$E\delta$ be the resulting configuration.

  Now consider the execution in which $p$ completes its ``critical''
  write, then $p'$ runs solo until it covers a register $r'''\neq r$.
  Let $E\delta'$ be the resulting configuration.
 Note that as the states of the memory in $E\delta$ and $E\delta'$ are
  identical, we have $I(\{E\delta,E\delta'\},\Pi\setminus\{p'\})$.  
Note that $\delta$ and $\delta'$ are $P$-only executions,
and that $p'$ covers $r$ in $E\delta$ and $r'''$ in~$E\delta'$.
\begin{enumerate}
\item If $r'''\in S$, as we have
  $\mathit{Confused}(P,(S\cup\{r'\})\setminus\{r\}, \E')$, 
  applying Lemma~\ref{lem:extendConfusion}, we obtain 
  $\mathit{Confused}(P\setminus\{p'\},(S\cup\{r'\}),(\E'\delta)\cup(\E'\delta'))$.
\item If $r'''\in\mathcal{R}\setminus S$, as we have $\mathit{Confused}(P,S,\E)$,
applying Lemma~\ref{lem:extendConfusion}, we obtain 
$\mathit{Confused}(P\setminus\{p'\},(S\cup\{r'''\}),(\E\delta)\cup(\E\delta'))$.
\end{enumerate} 
\item Otherwise, no write of $p$ is ``critical'', and we let it run
  from $E$ until it covers $r$ (recall that, as $p'$ covers $r''\in
  (S\cup\{r'\})\setminus\{r\}$, $p$ must eventually write to all
  registers in~$S$ or~$\{r,r'\}$ and, thus, to $r$).       
Let then $p'$ run until it covers $r$, as $p$, and let $\delta$ be this execution. 
From $E\delta$, let $p'$ run until it becomes poised to write to a register~$r'''\neq r$,
and then let $p$ perform its pending write on $r$. Let $\lambda$ be this extension.
Note that as $p'$ is hidden in $\lambda$, we have 
$I(\{E\delta\lambda, E\delta\lambda|_{\{p'\}}\},\Pi\setminus\{p\})$.
Note also that $\delta\lambda$ and $\delta\lambda|_{\{p\}}$ are $P$-only executions such that 
$p'$ covers $r$ in $E\delta\lambda|_{\{p\}}$ and covers $r'''$ in $E\delta\lambda$.
\begin{enumerate}
\item If $r'''\in S$, as we have $\mathit{Confused}(P,(S\cup\{r'\})\setminus\{r\}, \E')$,
applying Lemma~\ref{lem:extendConfusion}, we obtain
$\mathit{Confused}(P\setminus\{p'\},(S\cup\{r'\}),(\E'\delta\lambda)\cup(\E'\delta\lambda|_{\{p\}}))$.
\item If $r'''\in\mathcal{R}\setminus S$, as we have $\mathit{Confused}(P,S,\E)$,
applying Lemma~\ref{lem:extendConfusion}, we obtain
$\mathit{Confused}(P\setminus\{p'\},(S\cup\{r'''\}),(\E\delta\lambda)\cup(\E\delta\lambda|_{\{p\}}))$.
\end{enumerate} 
\end{itemize}
\end{enumerate}
\end{proof}

Our lower bound directly follows from
Lemmata~\ref{lem:initConfusion} and~\ref{lem:inductiveConfusion}: 

\begin{theorem}
Any $n$-process comparison-based $2$-obstruction-free SWMR memory implementation  
requires $n+1$ MWMR registers.
\end{theorem}

\begin{proof}
 By contradiction, suppose that an $n$-register algorithm exists.  
We show, by induction, that there is a reachable configuration in
which a process is confused on \emph{all} registers.
Lemma~\ref{lem:initConfusion} shows that there exists  a reachable configuration 
in which $n-1$ processes are confused on two registers.
We can therefore apply Lemma~\ref{lem:inductiveConfusion} and 
obtain a configuration with a confusion with strictly more registers. 
By induction, there exist then a set of configurations~$\D$
and $p\in \Pi$ such that $\mathit{Confused}(\{p\},\mathcal{R},\D)$.

Thus, any strict subset of $\mathcal{R}$ is covered by the remaining $n-1$ processes in some 
configuration in  the (indistinguishable for $p$) set of configurations $\D$. 
But $p$ may complete a Write operation
if and only its write value is present in a register which is 
not covered (by a process not aware of the value) in any of the configurations 
indistinguishable to $p$. Therefore, in an infinite solo execution,
$p$ must write infinitely often to \emph{all} registers. But then, any arbitrarily long 
execution by any other process can be hidden by incorporating 
sufficiently many steps of $p$, violating $2$-obstruction-freedom---a
contradiction.
\end{proof}

\section{Concluding remarks}
\label{sec:disc}

This paper shows that the optimal space complexity of SWMR
implementations depends on the desired progress condition: lock-free
algorithms trivially require $n$ registers, while $2$-obstruction-free ones
(and, thus, also $2$-lock-free ones) require $n+1$ registers.
We also extend the upper bound to $k$-lock-freedom, for all
$k=1,\ldots,n$, by presenting a $k$-lock-free SWMR implementation
using  $n+k-1$ registers.
A natural conjecture is that the algorithm is optimal, i.e., no such algorithm exists
for $n+k-2$ registers for all $k=1,\ldots,n$. 
Since for $k=1$, $2$ and $n$, 
$k$-obstruction-freedom and $k$-lock-freedom impose the same space complexity,
it also appears natural to expect that this is also true for all $k=1,\ldots,n$.

An interesting corollary to our results is that to implement a
$2$-obstruction-free SWMR memory we need strictly more space 
than to implement a $1$-lock-free one. But the two properties are,
in general, incomparable: a $2$-solo run in which only one process
makes progress satisfies $1$-lock-freedom, but not
$2$-obstruction-freedom, and a run in which $3$ or more processes are
correct but no progress is made satisfies $2$-obstruction-freedom, but not
$1$-lock-freedom.
The relative costs of incomparable progress properties, e.g., in the
$(\ell,k)$-freedom spectrum~\cite{BG15}, are yet to be understood. 

An SWMR memory can be viewed as a \emph{stable-set} abstraction
with a conventional \emph{put/get} interface: 
every participating process can put values to the set and 
get the set's content, and every get operation 
returns the values previously put.
For the stable-set abstraction, we can extend our results to the
\emph{anonymous} setting, where processes are not provided with unique
identifiers.
Indeed, we claim that the same algorithm may apply to the stable-set abstraction 
for anonymous systems when the number of participating processes $n$ is known.
But the question of whether an \emph{adaptive} solution exists (expressed
differently,  a solution that does not assume any upper bound 
on the number of participating processes) for anonymous systems remains open.

\bibliographystyle{plain}

\end{document}